\documentclass[aps,prd,nofootinbib,preprint,tightenlines]{revtex4-1}

\usepackage{todonotes,soul}
\usepackage{todonotes}
\usepackage{epsfig}
\usepackage{braket}
\usepackage{graphicx}
\usepackage{amsmath}
\usepackage{amsthm}
\usepackage{amssymb}
\usepackage{amsfonts}
\usepackage{mathbbol}
\usepackage{bm}
\usepackage{xcolor}
\usepackage{mathrsfs}

\newcommand{\bfk}{\mathbf{k}}

\newcommand{\be}{\begin{equation}}
\newcommand{\ee}{\end{equation}}
\newcommand{\bea}{\begin{eqnarray}}
\newcommand{\eea}{\end{eqnarray}}
\def\bml{\begin{subequations}}
\def\blea{\bml\begin{eqnarray}}
\def\eml{\end{subequations}}
\def\elea{\end{eqnarray}\eml}

\renewcommand{\sec}{Sec$.$~}

\newtheorem{theorem}{Theorem}

\newtheorem{corollary}[theorem]{Corollary}

\def\fmax{\phi_{\text{max}}}

\begin{document}
	\title{A singularity theorem for Einstein--Klein--Gordon theory}
	
	\author{Peter J. Brown}
	\author{Christopher J. Fewster}
	\author{Eleni-Alexandra Kontou}
	\affiliation{Department of Mathematics, University of York, Heslington, York YO10 5DD, United Kingdom}

	\begin{abstract}

		Hawking's singularity theorem concerns matter obeying the strong energy condition (SEC), which means that all observers experience a nonnegative effective energy density (EED), thereby guaranteeing the timelike convergence property. However, there are models that do not satisfy the SEC and therefore lie outside the scope of Hawking's hypotheses, an important example being the massive Klein--Gordon field.
		Here we derive lower bounds on local averages of the EED for solutions to the Klein--Gordon equation, allowing nonzero mass and nonminimal coupling to the scalar curvature. The averages are taken along timelike geodesics or over spacetime volumes, and our bounds are valid for a range of coupling constants
		including both minimal and conformal coupling. Using methods developed by Fewster and Galloway, 
		these lower bounds are applied 
		to prove a Hawking-type singularity theorem for solutions to the Einstein--Klein--Gordon theory, asserting that solutions with sufficient initial contraction at a compact Cauchy surface will be future timelike geodesically incomplete.\\  
		
		\noindent\emph{Dedicated to the memory of S.W.~Hawking}

	\end{abstract}

	\maketitle
	
	\noindent

	\section{Introduction}
	\label{sec:introduction}
	
	The conditions under which cosmological models either originate or terminate in a singularity provided an active subject of research in the decades prior to the breakthroughs made by Penrose~\cite{Penrose_prl:1965} and Hawking~\cite{Hawking:1966sx}. Results from that era mainly concern solutions with symmetries, as represented by the survey~\cite{HeckmannSchuecking:1962}. Raychaudhuri's work in 1955 represented a decisive step forward, because he was able to analyse inhomogeneous models using (a forerunner of) the equation that now carries his name. In their general
	and modern form~\cite{Ehlers:1993} the Raychaudhuri equations present the
	evolution of timelike geodesic congruences in a physically transparent fashion. For the special case of
	an irrotational congruence with velocity field $U^\mu$, the expansion $\theta=\nabla_\mu U^\mu$ satisfies
	\be
	\label{eqn:ray0}
	\nabla_U\theta =R_{\mu \nu} U^\mu U^\nu-2\sigma^2 -\frac{\theta^2}{n-1}  \,,
	\ee
	where $n$ is the spacetime dimension, $\sigma$ is the shear scalar and $R_{\mu\nu}$ is the Ricci tensor. Assuming that the 
	geometry is a solution to the Einstein equations
	\be \label{eqn:Einstein}
	G_{\mu\nu} = - 8\pi T_{\mu\nu},
	\ee
	the Raychaudhuri equation \eqref{eqn:ray0} becomes
	\be
	\label{eqn:ray1}
	\nabla_U\theta=-8\pi\rho -2\sigma^2 -\frac{\theta^2}{n-1}  \,,
	\ee
	where
	\begin{equation}\label{eq:EED_def}
		\rho=  T_{\mu \nu} U^\mu U^\nu-\frac{T}{n-2}  
	\end{equation}
	and $T=T^{\mu}_{\phantom{\mu}\mu}$. The quantity $\rho$ has appeared in general relativity
	since the works of Whittaker~\cite{Whittaker:1935} and Synge~\cite{Synge:1937}, playing
	the role of the mass-energy density in general relativistic versions of the Gauss law; Pirani~\cite{Pirani:2009} likewise identifies it as the `effective density of gravitational mass'.
	Here, imputing units of energy rather than mass, we will use the term \emph{effective energy density} (EED) for $\rho$.  
	Evidently, the sign of $\rho$ is crucial. If $\rho\ge 0$, that is, if the strong energy condition (SEC) holds, then the right-hand side of \eqref{eqn:ray1} is negative, driving $\theta\to -\infty$ in finite proper time. This is incompatible with geodesic completeness and implies the existence of a singularity.  
	
	Senovilla~\cite{Senovilla:2014} has described the skeleton of the singularity theorems in terms of a `pattern theorem' with three ingredients. An \emph{energy condition} establishes a focussing effect for geodesics, while a \emph{causality condition} removes the possibility of closed timelike curves and a \emph{boundary or initial condition} establishes the existence of some trapped region of spacetime. The goal of the singularity theorems is to show that the spacetime contains at least one incomplete causal geodesic; 
	we will divide singularity theorems into `Hawking-type' and `Penrose-type',
	depending on whether they demonstrate timelike or null geodesic incompleteness respectively. While Hawking-type results are based on the SEC, Penrose-type results assume the null energy condition (NEC), $T_{\mu\nu}k^\mu k^\nu\ge 0$ for all null $k^\mu$. 
	
	Hawking wrote that `[the energy conditions] are properties that any normal matter should have'~\cite[\S 5]{Hawking:1966sx} and indeed many models do respect the SEC. However, not all do, and in fact the massive minimally coupled Klein--Gordon field obeying $(\Box+m^2)\phi=0$ has EED
	\begin{equation}
		\rho = (\nabla_U \phi)^2 -\frac{m^2\phi^2}{n-2},
	\end{equation}
	which is easily made negative at individual points. Similarly, it is easily seen that the SEC and NEC fail for the nonminimally coupled Klein--Gordon field. This situation is exacerbated in quantum field theory, in which none of the pointwise energy conditions can hold~\cite{Epstein:1965zza}. We refer the reader to recent reviews of energy conditions \cite{Curiel:2017,MartinMorunoVisser:2017}.
	
	For these reasons there has long been interest in establishing singularity theorems under weakened energy assumptions. Examples include \cite{Tipler:1978zz,ChiconeEhrlich:1980,Borde:1987b, Roman:1988vv,Wald:1991xn,Borde:1994}, in which various averages of the energy density or related quantities are
	required to be nonnegative if the average is taken over a sufficiently large portion of a (half-)complete causal geodesic, or at least is intermittently nonnegative~\cite{Borde:1987b}.  
	Our approach in this paper follows~\cite{Fewster:2010gm}, in which (generalising results from~\cite{Galloway:1981}) it was shown among other things that suitable lower bounds on local weighted averages of $\rho$ are sufficient to derive singularity theorems of Hawking and Penrose type, even if $\rho$ is not everywhere positive or has a negative long-term average.
	
	The bounds adopted in~\cite{Fewster:2010gm} were inspired by the Quantum Energy Inequalities (QEIs) that have been established in various models of quantum field theory (see~\cite{Fewster2017QEIs} for a recent review). However, there is a significant gap between the results of~\cite{Fewster:2010gm} and a semiclassical Hawking-type singularity theorem, because there is so far no QEI version of the SEC. The purpose of this paper is to show that the classical nonminimally coupled massive Klein--Gordon field obeys lower bounds on $\rho$ of the type considered in~\cite{Fewster:2010gm}. The general approach is parallel to methods used in~\cite{Fewster:2006ti} to obtain averaged versions of the weak and null energy conditions for the classical nonminimally coupled scalar field. Elsewhere, we will use our results to establish QEI analogues of the SEC (cf.~\cite{Fewster:2007ec}); here, we use them to derive a new Hawking-type singularity theorem for the Einstein--Klein--Gordon system. In a completely different direction, we mention that the methods of~\cite{Fewster:2010gm}, and therefore bounds of the type developed here, could be used in other problems in relativity. See, for example~\cite{Lesourd:2017}, in which a version of Hawking's area theorem is proved under weakened hypotheses. 
	
	The paper is structured as follows. In \sec\ref{sec:nonmfield} we recall the energy-momentum tensor for the non-minimally coupled scalar field and the manner in which it can violate the pointwise SEC.  Next, in \sec\ref{sec:worldline}, we consider local averages of $\rho$ of the form
	\begin{equation}
		\int\rho(\gamma(\tau)) f(\tau)^2 \,d\tau,
	\end{equation}
	where $\gamma$ is a timelike geodesic parameterised by proper time and $f$ is a real-valued smooth and compactly supported function. Here, it not is assumed that the background spacetime and field together solve the Einstein--Klein--Gordon equations.
	We derive lower bounds on such averages that depend only the values of $\phi$, but not its derivatives. The bounds also depend on $\gamma$ and $f$ together with its derivatives and are valid for all values of the coupling $\xi$ in the interval $[0,2\xi_c]$ where $\xi_c$ is the conformal coupling constant ($\xi_c=1/6$ for $n=4$). We investigate the behaviour of these lower bounds under scaling 
	of $f$ and also derive constraints on the time for which $\rho$ can be more negative than some given value.  Section~\ref{sec:worldvolume} addresses similar questions for worldvolume averages of $\rho$ obtaining bounds valid on an interval containing $[0,\xi_c]$ for dimensions $n\ge 4$. In the special case of flat spacetime, one
	may prove that the average value of $\rho$ over all spacetime is nonnegative. 
	In \sec\ref{sec:ekg}, we return to worldline bounds, now adapted to the special case of solutions to the Einstein--Klein--Gordon system and obtaining a slightly refined bound, which is used in our discussion of singularity theorems in \sec\ref{sec:singularity}. There, we first establish a Hawking-type singularity theorem using methods taken from~\cite{Fewster:2010gm} and apply it to the Einstein--Klein--Gordon theory using our worldline bounds. This provides an 
	analogue to the Penrose-type singularity theorem for the nonminimally coupled 
	scalar field discussed in~\cite{Fewster:2010gm}. Finally, we conclude in \sec~\ref{sec:disc} with 
	a discussion of the magnitude of the initial contraction needed to ensure
	timelike geodesic incompleteness according to our results. 
	
	Our sign conventions are the $[-,-,-]$ of Misner, Thorne and Wheeler \cite{MTW}.  We write the d'Alembertian with respect to the metric $g$ as $\Box_g = g^{\mu\nu}\nabla_\mu\nabla_\nu$ and work in $n$ spacetime dimensions unless otherwise stated. Except in \sec\ref{sec:disc} we adopt units in which $G=c=1$. 
	
	\section{The non-minimally coupled field}
	\label{sec:nonmfield}
	
	The field equation for non-minimally coupled scalar fields is
	\be \label{eqn:field}
	(\Box_g+m^2+\xi R)\phi=0 \,,
	\ee
	where $\xi$ is the coupling constant and $R$ is the Ricci scalar. 
	The constant $m$ has dimensions of inverse length, which would be the
	inverse Compton wavelength if one regarded~\eqref{eqn:field} as the
	starting-point for a quantum field theory with massive particles.
	The Lagrangian is
	\be
	\label{eqn:lagrangian}
	L[\phi]=\frac{1}{2} [(\nabla \phi)^2-(m^2+\xi R)\phi^2 ] \,,
	\ee
	from which the stress energy tensor is obtained by varying the action with respect to the metric, giving
	\be
	\label{eqn:tmunu}
	T_{\mu \nu}=(\nabla_\mu \phi)(\nabla_\nu \phi)+\frac{1}{2} g_{\mu \nu} (m^2 \phi^2-(\nabla \phi)^2)+\xi(g_{\mu \nu} \Box_g-\nabla_\mu \nabla_\nu-G_{\mu \nu}) \phi^2 \,,
	\ee
	where $G_{\mu \nu}$ is the Einstein tensor.  The trace of the stress-energy tensor is given by
	\be
	\label{eqn:trace}
	T = \left(1-\frac{n}{2}\right) (\nabla\phi)^2 + \frac{n}{2}m^2\phi^2 + \xi\left((n-1)\Box_g -  \left(1-\frac{n}{2}\right)R \right)\phi^2 \,.
	\ee
	We should observe here that the field equation and the Lagrangian reduce to those of minimal coupling for flat spacetimes but the stress energy tensor does not. The effective energy density $\rho$ of
	Eq.~\eqref{eq:EED_def}  obtained from the stress-energy tensor Eq.~(\ref{eqn:tmunu}) for a timelike observer with $4$-velocity $U^\mu$ is
	\bea \label{eqn:SED}
	\rho &=&  (1-2\xi) U^\mu U^\nu (\nabla_\mu \phi)(\nabla_\nu \phi)-\frac{1-2\xi}{n-2}  m^2 \phi^2 -\frac{2\xi}{n-2} (\nabla \phi)^2  \nonumber\\
	&&\quad -2\xi U^\mu U^\nu \phi \nabla_\mu \nabla_\nu  \phi -\xi U^\mu U^\nu R_{\mu \nu} \phi^2+\frac{2\xi^2}{n-2}  R \phi^2-\frac{2\xi}{n-2} (\phi P_\xi \phi)    \,,
	\eea
	where $P_\xi=\Box_g+m^2+\xi R$ is the Klein-Gordon operator. The last term can be discarded ``on shell''  i.e. for $\phi$ satisfying  Eq.~(\ref{eqn:field}). For $\xi=0$ the EED further reduces to
	\be \label{eqn:minSED}
	\rho= U^\mu U^\nu (\nabla_\mu \phi)(\nabla_\nu \phi)-\frac{1}{n-2} m^2 \phi^2  \,.
	\ee
	From Eq.~\eqref{eqn:minSED} we can see that, even for minimally coupled fields, we find a violation of the SEC at any point in spacetime at which $ m^2\phi^2 \geq (n-2) (\nabla_U\phi)^2$, and the violation can be made arbitrarily large if $m$ or $\phi$ can be made large. We also observe a guaranteed violation whenever the field derivatives vanish, as we are left with a manifestly negative term.

\section{Worldline strong energy inequality}
\label{sec:worldline}

We will study the EED  of the stress-energy tensor of Eq.~\eqref{eqn:tmunu} with respect to freely falling observers. Let $\gamma$ be a a timelike geodesic parametrised by proper time $\tau$.
Let $f$ be a real-valued and compactly supported function $f \in C_0^2(\mathbb{R})$. We are interested in expressions of the form
\be \label{eqn:clASED}
\int_\gamma d\tau \,\rho \, f^2(\tau)=\int_\gamma d\tau \left(  T_{\mu \nu}\dot{\gamma}^\mu \dot{\gamma}^\nu-\frac{1}{n-2} T  \right)f^2(\tau) \,.
\ee
Eq.~(\ref{eqn:clASED}) ``on shell" reduces to
\bea\label{eqn:onshell-clASED}
\int_\gamma d\tau \,\rho\, f^2(\tau)&=&\int_\gamma d\tau \bigg((1-2\xi) (\nabla_{\dot{\gamma}}\phi)^2-\frac{1-2\xi}{n-2}m^2 \phi^2 -\frac{2\xi}{n-2} (\nabla \phi)^2 \nonumber\\
&& \qquad \qquad \qquad \qquad -2\xi \phi (\nabla_{\dot{\gamma}}^2 \phi)-\xi \dot{\gamma}^\mu \dot{\gamma}^\nu R_{\mu \nu} \phi^2+\frac{2\xi^2}{n-2} R \phi^2 \bigg) f^2(\tau) \,.
\eea
From Eq.~(9) of Ref.~\cite{Fewster:2006ti} we have
\be \label{eqn:posdif}
-2\xi \int_\gamma d\tau\, f^2(\tau) \phi  \nabla_{\dot{\gamma}}^2 \phi=2 \xi \int_\gamma d\tau  [\nabla_{\dot{\gamma}} (f(\tau) \phi)]^2-2\xi  \int_\gamma d\tau \phi^2 ( f'(\tau))^2 \,,
\ee
which is a difference of positive terms for positive coupling constant. Additionally we can write
\be
(1-2\xi) (\nabla_{\dot{\gamma}}\phi)^2-\frac{2\xi}{n-2} (\nabla \phi)^2=\left(1-\frac{\xi}{2\xi_c}\right) (\nabla_{\dot{\gamma}}\phi)^2+\frac{2\xi}{n-2} h^{\mu \nu} \nabla_\mu \phi \nabla_{\nu} \phi \,,
\ee
where $h^{\mu \nu}=\dot{\gamma}^\mu \dot{\gamma}^\nu-g^{\mu \nu}$ is a positive definite metric and $\xi_c$ is the conformal coupling constant defined as
\be
\xi_c=\frac{n-2}{4(n-1)} \,.
\ee
Applying the previous two identities to Eq.~\eqref{eqn:onshell-clASED}, we find that all the curvature independent terms are either positive or negative for $\xi \in[0,2 \xi_c]$. As a result, we have proved the following theorem.

\begin{theorem} 
\label{the:clline}
Let $\gamma$ be a timelike geodesic parametrized by proper time $\tau$ in $(M,g)$, where $M$ is a manifold with dimension $n\geq 2$. Let $T_{\mu \nu}$ be the stress-energy tensor of a scalar field with coupling constant $\xi \in[0,2\xi_c]$ and $f$ a real valued function of compact support. Then ``on shell''
\bea
\label{eqn:cline}
\int_\gamma d\tau \, \rho\, f^2(\tau)&\geq& - \int_\gamma d\tau \bigg\{ \frac{1-2\xi}{n-2} m^2 f^2(\tau)+\xi \bigg(2( f'(\tau))^2+R_{\mu \nu} \dot{\gamma}^\mu \dot{\gamma}^\nu f^2(\tau) \nonumber\\
&& \qquad \qquad  - \frac{2\xi}{n-2} Rf^2(\tau) \bigg) \bigg\}\phi^2  \,.
\eea
\end{theorem}
In fact, we have proved a slightly stronger bound that will be useful later on, namely
\bea
\label{eqn:cline2}
\int_\gamma d\tau \, \rho\, f^2(\tau)&\geq& - \int_\gamma d\tau \bigg\{ \frac{1-2\xi}{n-2} m^2 f^2(\tau)+\xi \bigg(2( f'(\tau))^2+R_{\mu \nu} \dot{\gamma}^\mu \dot{\gamma}^\nu f^2(\tau) \nonumber\\
&& \qquad  - \frac{2\xi}{n-2} Rf^2(\tau) \bigg) \bigg\}\phi^2
+ \int_\gamma d\tau 
\left(1-\frac{\xi}{2\xi_c}\right) (\nabla_{\dot{\gamma}}\phi)^2 f^2(\tau)
\,.
\eea
However \eqref{eqn:cline} has the advantage that only the field $\phi$, and not its derivative, appears on the right-hand side.

For flat spacetimes the bound of Theorem \ref{the:clline} becomes 
\be \label{eqn:clasflat}
\int_\gamma d\tau \, \rho\, f^2(\tau) \geq - \int_\gamma d\tau \bigg\{ \frac{1-2\xi}{n-2} m^2 f^2(\tau) +2 \xi ( f'(\tau))^2 \bigg\}\phi^2 \,,
\ee
while for minimally coupled fields, regardless of the curvature
\be
\int_\gamma d\tau \, \rho\, f^2(\tau) \geq -  \frac{1}{n-2}m^2  \int_\gamma d\tau \, f^2(\tau) \phi^2 \,.
\ee
In order to understand the significance of these results, it is useful to discuss 
some consequences of the flat spacetime bound Eq.~\eqref{eqn:clasflat}. First, 
let us consider its behaviour under rescaling of the smearing function $f$. 
Writing $\fmax$ for the maximum field amplitude of the field along the inertial trajectory $\gamma$,
\be
\label{eqn:phimax}
\fmax=\sup_\gamma |\phi | \,,
\ee
Eq.~\eqref{eqn:clasflat} implies 
\be
\int_\gamma d\tau \, \rho\, f^2(\tau) \geq - \fmax^2 \int_\gamma d\tau \bigg\{ \frac{1-2\xi}{n-2} m^2 f^2(\tau) +2 \xi ( f'(\tau))^2 \bigg\} \,
\ee
for any compactly supported real-valued $f$. Let us now assume that $f$ has unit $L^2$-norm. Introducing the rescaled function 
\be
f_{\lambda}(\tau)=\frac{f(\tau/\lambda)}{\sqrt{\lambda}} \,,
\ee
chosen so that its normalization is independent of the choice of $\lambda>0$
\be
\int d\tau f^2_\lambda (\tau)= \int d\tau f^2 (\tau)=1 \,,
\ee
we can write
\be
\int_\gamma d\tau \, \rho \, f^2_\lambda (\tau) \geq -  \fmax^2 \int_\gamma d\tau \frac{1}{\lambda} \bigg\{ \frac{1-2\xi}{n-2} m^2 f^2(\tau/\lambda) +\frac{2\xi}{\lambda^2} (f'(\tau/\lambda))^2 \bigg\}   \,.
\ee
Changing variables to $\tau \to \tau \lambda$ on the right-hand side and taking the limit $\lambda \to \infty$ we get
\be
\liminf_{\lambda \to \infty} \int_\gamma d\tau \, \rho \, f^2_\lambda (\tau)  \geq -  \frac{1-2\xi}{n-2} m^2 \fmax^2 \,.
\ee
This result may be interpreted as providing a lower bound on the long-term average value of $\rho$, which leaves open the possibility that the long-term average in the case $m>0$ can be negative, even for $\xi=0$. This can be contrasted with analogous results for the null energy condition in~\cite{Fewster:2006ti}, which establish ANEC in an appropriate limit.

A slightly different approach is to estimate the supremum of the EED over an open interval $I$ of proper time with duration $\tau_0$. This gives 
\be\label{eq:longterm}
\sup_{\gamma(I)} \rho \geq - \left\{ \frac{1-2\xi}{n-2}m^2+\frac{2\xi \pi^2}{\tau_0^2} \right\} \sup_{\gamma(I)} |\phi|^2\,, 
\ee
where use the fact that 
\begin{equation}
\inf_f \frac{\|f'\|^2}{\|f\|^2}= \frac{\pi^2}{\tau_0^2}
\end{equation} 
where $\|\cdot\|$ denotes the $L^2$-norm and the infimum is 
taken over all smooth $f$ with compact support in an interval of length $\tau_0$ (see Ref.~\cite{Fewster:1998xn,Fewster:2006ti} for similar arguments). Thus violations of the 
SEC beyond the level of the long term average bound in Eq.~\eqref{eq:longterm} are possible only on timescales $\tau_0\ll \xi^{1/2}m^{-1}$, and not at all if $\xi=0$.

\section{Worldvolume strong energy inequality}
\label{sec:worldvolume}

Instead of averaging the EED over a worldline we can average over spacetime volumes. Let $U^\mu$ be a future-directed timelike unit vector field. Introducing $f(x)$ as a smearing function with compact support, and writing $V^\mu = f(x)U^\mu$, the averaged EED for the nonminimally coupled scalar field is ``on shell"
\bea \label{eqn:ASEDvol}
&& \int dVol \, \rho\, f^2(x)=\int dVol  \bigg\{ (1-2\xi) f^2(x) U^\mu U^\nu (\nabla_\mu \phi)(\nabla_\nu \phi) -\frac{1-2\xi}{n-2}  m^2 \phi^2 f^2(x)   \nonumber\\
&&\qquad \qquad -\frac{2\xi}{n-2} (\nabla \phi)^2 f^2(x)-\xi V^\mu V^\nu (2\phi \nabla_\mu \nabla_\mu \phi+ R_{\mu \nu} \phi^2)+\frac{2\xi^2}{n-2} R\phi^2 f^2(x) \bigg\} \,,
\eea
where $\rho$ is given by Eq.~\eqref{eqn:SED}. From Eq.~(34) of Ref.~\cite{Fewster:2006ti} we have
\bea \label{eqn:posdifvol}
-\xi \int dVol \, V^\mu V^\nu( 2\phi \nabla_\mu \nabla_\nu \phi+R_{\mu \nu} \phi^2)&& =2\xi \int dVol \, [\nabla_\mu (V^\mu \phi)]^2 \nonumber\\
&&-\xi \int dVol \, [(\nabla_\mu V^\mu)^2+(\nabla_\mu V^\nu)(\nabla_\nu V^\mu)] \phi^2 \,,
\eea
which is a generalization of Eq.~(\ref{eqn:posdif}) that was used for the worldline average. We can also write
\bea \label{eqn:posmetr}
(1-2\xi) U^\mu U^\nu (\nabla_\mu \phi)(\nabla_\nu \phi)-\frac{2\xi}{n-2} (\nabla \phi)^2&=&\left(1-2\xi\frac{n-1}{n-2}\right) U^\mu U^\nu (\nabla_\mu \phi)(\nabla_\nu \phi)\nonumber \\
&&\qquad \qquad +\frac{2\xi}{n-2} h^{\mu \nu} (\nabla_\mu \phi) (\nabla_\nu \phi)  
\,,
\eea
where $h^{\mu \nu}=U^\mu U^\nu - g^{\mu \nu}$ is a positive definite metric. Now all curvature-independent terms are either positive or negative for $\xi \in[0,2\xi_c]$, and we have the following bound for the averaged EED
\bea \label{eqn:volbound1}
\int dVol \, \rho\, f^2(x) &\geq& -\int dVol \bigg\{ \frac{1-2\xi}{n-2} m^2 f^2(x)+\xi  [(\nabla_\mu V^\mu)^2+(\nabla_\mu V^\nu)(\nabla_\nu V^\mu)] \nonumber\\
&& \qquad \qquad \qquad -\frac{2\xi^2}{n-2} R f^2(x) \bigg\}  \phi^2 \,. 
\eea
This bound retains many features of the worldline bound Eq.~\eqref{eqn:cline}. In particular the mass-dependent term $(1-2\xi) m^2 f^2(x)/(n-2)$ appears in both. In the worldline case, this term
prevented us from showing that the long-term worldline average of $\rho$ is positive. For worldvolume averaging, however, we can use the field equation \eqref{eqn:field} along with successive integration-by-parts to derive an alternative bound that has no explicit mass-dependence and remains free from any field derivatives. This is achieved as follows. 
The field equation allows us to rewrite the stress-energy tensor as
	\be
	T_{\mu \nu}=(1-2\xi) (\nabla_\mu \phi)(\nabla_\nu \phi)-\frac{1}{2} (1-4\xi)( \phi \Box_g \phi+(\nabla \phi)^2)-2\xi \phi \nabla_\mu \nabla_\mu \phi-\xi R_{\mu \nu} \phi^2 +\frac{1}{2} g_{\mu \nu} (\phi P_\xi \phi)\,,
	\ee 
resulting in an alternative expression for the EED,
	\bea
	\label{eqn:SECbox}
	\rho &=& (1-2\xi) U^\mu U^\nu (\nabla_\mu \phi)(\nabla_\nu \phi) +\frac{1-2\xi}{n-2}(\phi \Box_g \phi)-\frac{2\xi}{n-2} (\nabla \phi)^2  \nonumber\\
	&& -2 \xi U^\mu U^\nu \phi \nabla_\mu \nabla_\nu \phi-\xi U^\mu U^\nu R_{\mu \nu} \phi^2+\frac{1}{n-2}\xi R \phi^2  -\frac{1}{n-2}(\phi P_\xi \phi)  \,.
	\eea
Thus, we may write the averaged EED for ``on shell" field configurations, as
\bea \label{eqn:ASED2}
&& \int dVol \, \rho\, f^2(x) =\int dVol  \bigg\{ (1-2\xi) U^\mu U^\nu (\nabla_\mu \phi)(\nabla_\nu \phi)f^2(x)+\frac{1-2\xi}{n-2} f^2(x) \phi \Box_g \phi \nonumber\\
&&\qquad \qquad -\frac{2\xi}{n-2} (\nabla\phi)^2 f^2(x)-\xi V^\mu V^\nu( 2\phi \nabla_\mu \nabla_\nu \phi+R_{\mu \nu} \phi^2)+\frac{1}{n-2} \xi R\phi^2 f^2(x) \bigg\} \,.
\eea
Writing
\be
\phi \Box_g \phi=\frac{1}{2} \Box_g \phi^2-g^{\mu \nu} (\nabla_\mu \phi)(\nabla_\nu \phi) \,,
\ee
the EED becomes
\bea \label{eqn:ASED3}
\int dVol \, \rho\ f^2(x) &=&\int dVol  \bigg\{ \left( \frac{n-3}{n-2}-2\xi \right) U^\mu U^\nu (\nabla_\mu \phi)(\nabla_\nu \phi)f^2(x) \nonumber\\
&&\qquad \qquad  +\frac{h^{\mu \nu}}{n-2}(\nabla_\mu \phi)(\nabla_\nu \phi)f^2(x)+\frac{1-2\xi}{2(n-2)}   (\Box_g \phi^2)f^2(x) \nonumber \\
&&\qquad \qquad-\xi V^\mu V^\nu (2\phi \nabla_\mu \nabla_\nu \phi+R_{\mu \nu} \phi^2)+\frac{\xi R}{n-2} \phi^2 f^2(x) \bigg\} \,. 
\eea
By integrating by parts we can rewrite the third term of the integral
\be
 \int d Vol \, \frac{1-2\xi}{2(n-2)}   (\Box_g \phi^2)f^2(x) = \frac{1-2\xi}{2(n-2)}  \int d Vol \,  (\Box_g f^2(x)) \phi^2 \,,
\ee
where we used the fact that the boundary terms vanish. Using Eq.~(\ref{eqn:posdifvol}) and discarding the positive terms from the bound for $\xi \in [0,\xi_v]$, where 
\be
\xi_v= \frac{n-3}{2(n-2)} \,,
\ee 
we can write
\bea \label{eqn:volbound2}
\int dVol \, \rho\, f^2(x)  &\geq& - \int dVol \bigg\{ -\frac{1-2\xi}{2(n-2)}(\Box_g f^2(x))+\xi  [(\nabla_\mu V^\mu)^2+(\nabla_\mu V^\nu)(\nabla_\nu V^\mu)]  \nonumber\\
&& \qquad \qquad  \qquad \qquad- \frac{1}{n-2} \xi R f^2(x) \bigg\} \phi^2 \,.
\eea
Note that $\xi_v < 2\xi_c$ for any spacetime dimension $n>2$, while $\xi_c<\xi_v$ for $n\ge 4$. Using Eqs.~(\ref{eqn:volbound1},\ref{eqn:volbound2}) we have proved following theorem:
\begin{theorem}
\label{the:clasvol}
If $M$ is a manifold with metric $g$ and dimension $n \geq 3$, $T_{\mu \nu}$ the stress-energy tensor of a scalar field with coupling constant $\xi \in [0,\xi_v]$ and $f$ a real valued function evaluated on a spacetime point $x$ then ``on shell ''  
\be
\int dVol \, \rho \, f^2(x)  \geq -\min \{ \mathcal{B}_1,\mathcal{B}_2\} \,,
\ee
where
\be
\mathcal{B}_1=\int dVol \bigg\{ \frac{1-2\xi}{n-2}  m^2 f^2(x) -\frac{2\xi^2 R}{n-2} f^2(x) + \xi  \left[(\nabla_\mu V^\mu)^2+(\nabla_\mu V^\nu)(\nabla_\nu V^\mu)\right] \bigg\}  \phi^2 \,,
\ee
and 
\be
\mathcal{B}_2=\int dVol \bigg\{ -\frac{1-2\xi}{2(n-2)}(\Box_g f^2(x)) - \frac{\xi R}{n-2} f^2(x)+\xi  [(\nabla_\mu V^\mu)^2+(\nabla_\mu V^\nu)(\nabla_\nu V^\mu)] \bigg\} \phi^2 \,.
\ee
\end{theorem}
Note that this result is restricted to minimal coupling in $n=3$.
For flat spacetimes the bounds of Theorem \ref{the:clasvol} become
\be \label{eqn:B1flat}
\mathcal{B}_1=\int dVol \left\{ \frac{1-2\xi}{n-2}  m^2 f^2(x) +\xi  [(\nabla_\mu V^\mu)^2+(\nabla_\mu V^\nu)(\nabla_\nu V^\mu)] \right\}  \phi^2 \,,
\ee
and
\be \label{eqn:B2flat}
\mathcal{B}_2=\int dVol \bigg\{ -\frac{1-2\xi}{2(n-2)} (\Box f^2(x)) + \xi  [(\nabla_\mu V^\mu)^2+(\nabla_\mu V^\nu)(\nabla_\nu V^\mu)] \bigg\} \phi^2 \,,
\ee
while for minimally coupled fields on any spacetime, 
\be
\mathcal{B}_1=\frac{1}{n-2} m^2  \int dVol \, f^2(x)   \phi^2 \,, \text{ and } \mathcal{B}_2= -\frac{1}{2(n-2)} \int dVol \, (\Box_g f^2(x)) \phi^2 \,.
\ee
We now investigate the behaviour of  Eq.~(\ref{eqn:B2flat}) under rescaling of the smearing function $f$. First let $\fmax$ be the maximum amplitude of the field 
\be
\fmax=\sup_M |\phi | \,,
\ee
so we can take it out of the bound, yielding
\be
\label{eqn:fmax}
\int dVol \, \rho\, f^2(x) \geq - \fmax^2 \int dVol \bigg\{ -\frac{1-2\xi}{2(n-2)}  (\Box_g f^2(x)) + \xi  [(\nabla_\mu V^\mu)^2+(\nabla_\mu V^\nu)(\nabla_\nu V^\mu)] \bigg\} \,.
\ee
(Eq.~\eqref{eqn:fmax} also holds if the supremum in the definition of $\phi_{\max}$ is taken over the support of $f$. However, in order to keep $\phi_{\max}$ constant for all rescaled smearings, we extend its definition to the entire manifold.)
Consider a translationally invariant unit timelike vector field $U^\mu$ and define the rescaled smearing function $f_\lambda$ for $\lambda>0$ to be
\be
f_\lambda(x)=\frac{f(x/\lambda)}{\lambda^{n/2}} \,,
\ee
so that its normalization is independent of the choice of $\lambda$
\be
\int dVol \, f^2_\lambda (x)= \int dVol \, f^2 (x)=1 \,.
\ee

Replacing $f$ by $f_\lambda$, the right-hand side of Eq.~(\ref{eqn:fmax}) becomes
\bea
&&-\frac{1}{4} \int dVol  \bigg( -\frac{1-2\xi}{2(n-2)}  \Box f^2_\lambda(x)+  \xi  [(U^\mu [\nabla_\mu f_\lambda (x)] )^2+(U^\nu [\nabla_\mu f_\lambda(x)])(U^\mu [\nabla_\nu f_\lambda(x)])] \bigg) \fmax^2 \nonumber\\
&&= -\frac{1}{4} \int dVol \frac{1}{\lambda^{n+2}} \bigg( -\frac{1-2\xi}{2(n-2)} \Box f^2(x/\lambda) +  \xi  [(U^\mu [\nabla_\mu f (x/\lambda)] )^2\\
&&\qquad \qquad \qquad \qquad \qquad \qquad  \qquad \qquad+(U^\nu [\nabla_\mu f(x/\lambda)])(U^\mu [\nabla_\nu f(x/\lambda)])] \bigg) \fmax^2 \nonumber \,,
\eea
where we used the fact that $U^\mu$ is translationally invariant and so its derivatives vanish. Changing variables $x \to \lambda x$ gives
\be
-\frac{1}{4} \int dVol \frac{1}{\lambda^2} \bigg( -\frac{1-2\xi}{2(n-2)} \Box f^2(x)+  \xi  [(U^\mu [\nabla_\mu f (x)] )^2+(U^\nu [\nabla_\mu f(x)])(U^\mu [\nabla_\nu f(x)])] \bigg) \fmax^2 \,.
\ee
In the limit of large $\lambda$ the bound goes to zero and we have
\be
\liminf_{\lambda \to \infty} \int dVol \, \rho\, f^2_\lambda(x) \geq 0 \,,
\ee
thus establishing an averaged SEC for flat spacetimes. A similar calculation for the $\mathcal{B}_1$ bound gives a weaker, negative, bound in this case. 

\section{A worldline inequality for the Einstein--Klein--Gordon system}
\label{sec:ekg}

The inequalities proved in Sections~\ref{sec:worldline} and~\ref{sec:worldvolume} 
are valid for solutions to the Klein--Gordon equation on an arbitrary fixed background spacetime.  
In this section we discuss how our worldline bound can be adapted to provide
more specific information about solutions to the full Einstein--Klein--Gordon system. 

In our discussion it will be important that the field magnitude is constrained below
a critical value. To see why, recall from Eq.~\eqref{eqn:tmunu} that the stress-energy tensor of the nonminimally coupled scalar field contains a term proportional to the Einstein tensor. Therefore the Einstein equations $G_{\mu\nu}=-8\pi T_{\mu\nu}$ can be rearranged into the form
\begin{equation}
G_{\mu\nu} =\frac{[\text{terms in $\phi$, $\nabla\phi$ and $\nabla\nabla\phi$}]_{\mu\nu}}{1-8\pi\xi\phi^2} ,
\end{equation}
where the numerator on the right-hand side no longer contains the Einstein tensor. 
For this reason, values of $|\phi|$ larger than the critical value $(8\pi\xi)^{-1/2}$ are considered unphysical since they correspond to a change of sign of the physical Newton's constant. See for example Ref.~\cite{Barcelo:2000zf}. 
 
We now adapt our worldline bounds of \sec~\ref{sec:worldline} to solutions of the Einstein--Klein--Gordon theory. Taking the trace of the Einstein equation $G_{\mu\nu}=-8\pi T_{\mu\nu}$ gives
\begin{equation}\label{eq:trEins}
\left(\frac{n}{2}-1\right)R = 8\pi T \,.
\end{equation}
This can be used, in combination with the Klein--Gordon equation, to rearrange the trace of the stress-energy tensor of the non-minimally coupled scalar field which is given by Eq.~\eqref{eqn:trace}
\begin{align}
(1-8\pi\xi\phi^2)T &= \left(1-\frac{n}{2}\right) (\nabla\phi)^2 + \frac{n}{2}m^2\phi^2 + \xi (n-1)\Box_g\phi^2 \nonumber \\
&=\left(1-\frac{n}{2}+2\xi(n-1)\right) (\nabla\phi)^2 + \frac{n}{2}m^2\phi^2 - 2\xi (n-1)(m^2+\xi R)\phi^2 
\end{align}
and using \eqref{eq:trEins} again, 
\begin{equation}
\left(1-8\pi\xi(1-\xi/\xi_c)\phi^2\right)T = -2(n-1)(\xi_c-\xi)  (\nabla\phi)^2 + \left(1+2(n-1)(\xi_c-\xi) \right)m^2\phi^2  \,.
\end{equation}
Therefore, if $\xi\le \xi_c$, 
\begin{equation}
\left(1-8\pi\xi(1-\xi/\xi_c)\phi^2\right)T \ge -2(n-1)(\xi_c-\xi)  (\nabla_{\dot{\gamma}}\phi)^2  \,,
\end{equation}
where we used the fact that $h^{\mu \nu}=\dot{\gamma}^\mu \dot{\gamma}^\nu-g^{\mu \nu}$ is a positive definite metric. If we take the maximum value of the field less than the critical value so $8\pi\xi\phi^2\le 1$, we have
\begin{equation}
\xi^2 R\phi^2 \ge -  (\xi_c-\xi)  (\nabla_{\dot{\gamma}}\phi)^2 \,,
\end{equation}
where we used Eq.~\eqref{eq:trEins}. 
Now we can replace the term including the Ricci scalar in the bound of Eq.~\eqref{eqn:cline2} using the inequality 
\begin{equation}\label{eqn:Ricci_ineq}
\int \frac{2\xi^2 R\phi^2}{n-2} f^2(\tau) d\tau \ge -\int \frac{2(\xi_c-\xi)}{n-2}  (\nabla_{\dot{\gamma}}\phi)^2  f^2(\tau) d\tau \,.
\end{equation}
This gives the following bound for any solution to the Einstein--Klein--Gordon system
\bea
\label{eqn:cline3}
\int_\gamma d\tau \, \rho\, f^2(\tau)&\geq& - \int_\gamma d\tau \bigg\{ \frac{1-2\xi}{n-2} m^2 f^2(\tau)+\xi \bigg(2( f'(\tau))^2+R_{\mu \nu} \dot{\gamma}^\mu \dot{\gamma}^\nu f^2(\tau) \bigg) \bigg\}\phi^2 \,,
\eea 
valid for $0\le \xi\le\xi_c$. To get the bound of Eq.~\eqref{eqn:cline3} from Eqs.~\eqref{eqn:cline2} and~\eqref{eqn:Ricci_ineq}, we discarded 
\begin{equation}
\int_\gamma \left(1-\frac{1}{2(n-1)}-2\xi\right) (\nabla_{\dot{\gamma}}\phi)^2  f^2(\tau) d\tau \geq \frac{1}{2} \int_\gamma (\nabla_{\dot{\gamma}}\phi)^2  f^2(\tau) d\tau \geq 0 \,,
\end{equation}
for $\xi\le\xi_c$. Now noticing that
\be
R_{\mu \nu} \dot{\gamma}^\mu \dot{\gamma}^\nu=-8\pi \rho \,,
\ee
we can move this term to the left side of the inequality of (\ref{eqn:cline3})
\be
\int_\gamma d\tau \, R_{\mu \nu} \dot{\gamma}^\mu \dot{\gamma}^\nu f^2(\tau) (1-8\pi \xi \phi^2) \leq   \int_\gamma d\tau \bigg\{ \frac{1-2\xi}{n-2} m^2 f^2(\tau)+2 \xi ( f'(\tau))^2 \bigg\}8\pi \phi^2 \,.
\ee
Since  $\phi$ is less than $(8\pi \xi)^{-1/2}$ we can absorb the factor $(1-8\pi \xi \phi^2)$ in $f(\tau)$ and state the following theorem
\begin{theorem}
	\label{the:ekgineq}
Suppose $(M,g,\phi)$ is a solution to the Einstein--Klein--Gordon equation in dimension $n>2$ with coupling constant $\xi\in[0,\xi_c]$ and  $|\phi| \le (8\pi\xi)^{-1/2}$. Let $\gamma$ be a timelike geodesic parametrized by proper time $\tau$ in $(M,g)$ and $f$ a real valued function. Then
\be
\label{eqn:ffactor}
\int_\gamma d\tau \, R_{\mu \nu} \dot{\gamma}^\mu \dot{\gamma}^\nu f^2(\tau)  \leq  \int_\gamma d\tau \bigg\{ \left(\frac{1-2\xi}{n-2}\right) \frac{m^2 f^2(\tau)}{1-8\pi \xi \phi^2}  +2 \xi  \left(\frac{d}{d\tau}\frac{f(\tau)}{\sqrt{1-8\pi \xi \phi^2}}\right)^2 \bigg\}8\pi \phi^2 \,. 
\ee
\end{theorem}

This inequality has the advantage that the left-hand side is geometric, while only non-geometric terms appear on the right-hand side. It will enable us to prove
a singularity theorem for this system.

\section{A Hawking-type singularity theorem}
\label{sec:singularity}

In this section we establish a Hawking-type singularity theorem with a weakened energy condition.  A similar Penrose-type singularity theorem was discussed in \cite{Fewster:2010gm}.\footnote{Note that \cite{Fewster:2010gm} employs $+++$ sign conventions.}
We then use the result of \sec\ref{sec:ekg} to obtain a Hawking-type singularity theorem for 
the non-minimally coupled Einstein--Klein--Gordon theory.  

\begin{theorem}\label{thm:sing}
	Let $(M,g)$ be a globally hyperbolic spacetime of dimension $n>2$, and let $S$ be a
	smooth compact spacelike Cauchy surface for $(M,g)$. Suppose that
	\begin{enumerate}\renewcommand{\theenumi}{\alph{enumi}}
		\item\label{it:lifetime} there exists $\tau_0>0$ such that the congruence of 
		future-directed unit-speed geodesics issuing orthogonally from $S$ can be continued to the past of $S$ for a proper time of at least $\tau_0$ with a smooth velocity field $U^\mu$ and expansion $\theta=\nabla_\mu U^\mu$;
		\item\label{it:Rbdsing} there are positive constants $Q$ and $\tilde{Q}$ such that, along each future complete unit speed timelike geodesic $\gamma:[-\tau_0,\infty) \to M$ issuing orthogonally from $S$ 
		one has an inequality
		\begin{equation}\label{eq:Rbdsing}
			\int R_{\mu \nu}\dot{\gamma}^\mu \dot{\gamma}^\nu f(\tau)^2\,d\tau \le Q(\|f'\|^2+ \tilde{Q}^2\|f\|^2),
		\end{equation}
		where $|| \cdot ||$ is the $L^2$-norm, for all smooth, real-valued $f$ compactly supported in $(-\tau_0,\infty)$;
		\item\label{it:singcon} for some $K>0$, (i) the inequality
		\begin{equation}\label{eq:singcon}
			\nabla_U \theta|_{\gamma(\tau)}+\frac{\theta(\gamma(\tau))^2}{n-1}\ge  Q(\tilde{Q}^2-K^2) \qquad \textrm{on $(-\tau_0,0]$}
		\end{equation}
		holds along every future-directed unit-speed geodesic $\gamma(\tau)$  
		issuing orthogonally from $S$ at $\tau=0$, and \\
		(ii) the expansion $\theta$ on $S$ obeys 
		\begin{equation}\label{eq:theta0bd}
			\theta|_S <-\tilde{Q}\sqrt{Q(n-1)+Q^2/2}-\frac{1}{2}Q K \coth{(K\tau_0)}\,.
		\end{equation} 
		Alternatively, 
		it is sufficient if (c)(i) holds with  
		\be
		\label{eqn:singconb}
		\nabla_U \theta|_{\gamma(\tau)} \geq 0 \qquad \textrm{on $(-\tau_0,0]$,}
		\ee 
	 	in place of \eqref{eq:singcon}, and (c)(ii) holds either as before (for some $K>0$) or with
		\eqref{eq:theta0bd} replaced by
		\be
		\label{eqn:theta0b}
		\theta|_S < -\tilde{Q} \sqrt{Q(n-1)+Q^2/2}-\frac{Q}{2 \tau_0} \,.
		\ee 
	\end{enumerate}
	Then the spacetime is future geodesically incomplete. 
\end{theorem}

{\noindent\em Remarks:} 
1. Note that hypotheses~\eqref{it:lifetime} and~\eqref{it:singcon} refer to the recent past of the Cauchy surface $S$, and therefore would in principle be amenable to observational confirmation. 
2. The proof shows that the expansion of the geodesic congruence normal to $S$ must actually diverge to $-\infty$ at finite time. From this perspective it may seem strange that~\eqref{eq:singcon} can be satisfied if $\dot{\theta}$ is large and positive on $(-\tau_0,0]$. However this is just an expression of the averaged bound~\eqref{eq:Rbdsing}: large positive values of $R_{\mu\nu}\dot{\gamma}^\mu\dot{\gamma}^\nu$ in the recent past must be counterbalanced by (even larger) negative values in the future (this follows by exactly the same reasoning used in explorations of `quantum interest'~\cite{Ford:1999qv, Fewster:1999kr}) which drive the expansion to $-\infty$. 3. The constants $Q$ and $\tilde{Q}$ can be global or be allowed to vary between geodesic congruences if that leads to a tighter bound. 
4. Clearly $\tau_0$ may be replaced in hypothesis (c) by any $\tilde{\tau}_0\in (0,\tau_0]$, giving useful additional freedom. Reducing $\tilde{\tau}_0$ means that \eqref{eq:singcon} can perhaps be satisfied with a smaller value of $K$, although this needs to be weighed against any consequent increase in $Q K \coth{(K\tilde{\tau}_0)}$. In any case there is an optimum value of $\tilde{\tau}_0$ for any fixed function $\theta$. 

\begin{proof}
	
	The beginning of the proof is exactly the same as the singularity theorem 5.1 in Ref.~\cite{Fewster:2010gm}. We suppose that the spacetime is timelike geodesically complete, and aim for a contradiction. General properties of globally hyperbolic spacetimes with compact Cauchy surfaces guarantee the existence of a future-directed
	\emph{$S$-ray} $\gamma$ --- that is, $\gamma$ is a unit-speed geodesic, issuing orthogonally from $S$, so that the Lorentzian distance from each $\gamma(\tau)$ to $S$ is precisely $\tau$, for all $\tau\in [0,\infty)$ --- $\gamma$ is complete by assumption. There is a neighbourhood of $\gamma$ in $J^+(S)$ in which the Lorentzian distance $\rho_S(p)$ from $p$ to $S$ is smooth. (We choose conventions so that $\rho_S$ is positive for timelike separation.) In this neighbourhood, the velocity field $U^\mu=\nabla^\mu\rho_S(p)$ is
	a smooth future-directed unit timelike vector field which is irrotational and orthogonal to $S$. 
	We now restrict to the geodesic $\gamma$ and write $\theta(\tau):=\nabla_\mu U^\mu|_{\gamma(\tau)}$ for the expansion along $\gamma$. By the above properties, $\theta(\tau)$ 
 	is a smooth solution to  Raychaudhuri's equation 
	\be
	\label{eqn:ray}
	\dot{\theta}(\tau)=r(\tau) -\frac{1}{n-1} \theta(\tau)^2 \,.
	\ee
 where $r(\tau)=R_{\mu \nu} \dot{\gamma}^\mu \dot{\gamma}^\nu-2\sigma(\gamma(\tau))^2$ and $\sigma$ is the shear scalar.  By our assumption (\ref{it:lifetime}) together with the assumption of future geodesic completeness, this equation is valid on $(-\tau_0,\infty)$.  Additionally note that if condition (b) holds then, 
 \be
 \label{eqn:rcondb}
  	\int r(\tau) f(\tau)^2\,d\tau \le Q(\|f'\|^2+ \tilde{Q}^2\|f\|^2) =: |||f|||^2\,,
  	\ee
  	is also true for all $f$. We proceed to prove that, contrary to what has just been shown,
  	Eq.~\eqref{eqn:ray} can have no smooth solution on $(-\tau_0,\infty)$ under these conditions; indeed, the solution must tend to $-\infty$ at finite positive time.
	
	Suppose first that $r(t)\equiv r$ is constant and note that \eqref{eq:theta0bd} or \eqref{eqn:theta0b} imply that $\theta_0=\theta(0)<0$. If $r<0$ then the unique
	solution to~\eqref{eqn:ray} is 
	\begin{equation}
	\theta(\tau)=-\sqrt{-(n-1)r} \cot{\left[\sqrt{\frac{-r}{n-1}}(\tau_*-\tau)\right]} \,, \quad \tau_*=\sqrt{\frac{n-1}{-r}} \cot^{-1}{\left[\frac{-\theta_0}{\sqrt{-(n-1)r}}\right]} \,,
	\end{equation}
	using the branch of arc-cotangent valued in $(0,\pi)$. As $\theta_0<0$, we have $\tau_*>0$ and the solution blows up as $\tau\to\tau_*$.
	Similarly, if $r=0$, the solution is
	\begin{equation}
		\theta(\tau)=\frac{n-1}{\tau-\tau^*}, \qquad \tau_*=-\frac{n-1}{\theta_0}
	\end{equation}
	and again there is blow-up as $\tau\to\tau_*>0$. If $r>0$, then  \eqref{eqn:rcondb} implies that
	$r\leq Q\tilde{Q}^2$. Using Eq.~\eqref{eq:theta0bd}, we have $\theta_0<-\tilde{Q}\sqrt{Q(n-1)}\leq -\sqrt{-(n-1)r}$
	and the solution is
	\begin{equation}
		\theta(\tau)=-\sqrt{(n-1)r} \coth{\left[\sqrt{\frac{r}{n-1}}(\tau_*-\tau)\right]} \,, \qquad \tau_*=\sqrt{\frac{n-1}{r}} \coth^{-1}{\left[\frac{-\theta_0}{\sqrt{(n-1)r}}\right]} \,,
	\end{equation}
	again blowing up at finite positive time. Therefore, Raychaudhuri's equation has no solution on $(-\tau_0,\infty)$ if $r(\tau)$ is constant, contradicting the existence of the solution shown above. Therefore $r(\tau)$ must be nonconstant.
	
	We may choose $\epsilon>0$ so that $\theta_0+\epsilon$ is also less than the right-hand side of
	\eqref{eq:theta0bd}. By Lemma~6.1 of~\cite{Fewster:2010gm}, there exists $c>0$ and a smooth function
	$h$ supported on $[-\tau_0,\infty)$ with $h(\tau)=e^{-c\tau/(n-1)}$ on $[0,\infty)$ and 
	so that the second inequality in
	\begin{equation}
	\label{eqn:thefour}
	-\theta_0 > 
	\epsilon + \tilde{Q}\sqrt{Q(n-1) + Q^2/2} + \frac{1}{2}QK\coth K\tau_0
	\geq \frac{c}{2} + |||h|||^2 - \int_{-\tau_0}^0 h^2(\tau) r(\tau) \,d\tau   \,.
	\end{equation}
	holds (the first one holds by virtue of our choice of $\epsilon$). 
	Eq.~\eqref{eqn:thefour} implies that \eqref{eqn:ray} has no solution on $[0,\infty)$ by Theorem~4.1 of~\cite{Fewster:2010gm} (applied with $z=-\theta$, $r_0\equiv 0$, $s=n-1$ and with the opposite sign convention for $r$). This contradicts the existence of the solution shown above and therefore demonstrates that the spacetime is future timelike geodesically incomplete.
	
	It remains to show that the alternative conditions stated in hypothesis (c) also suffice. 
	For \eqref{eqn:singconb}, observe that, together with \eqref{eq:theta0bd},
	it implies that $\theta \leq \theta_0<0$ on $(-\tau_0,0]$ and that
	\begin{equation}
		\dot{\theta} + \frac{\theta^2}{n-1} \geq \frac{\theta_0^2}{n-1} \geq  \frac{\tilde{Q}^2(Q(n-1) + Q^2/2)}{n-1} \geq Q\tilde{Q}^2 \qquad \textrm{on $(-\tau_0,0]$} 
	\end{equation}
	as a result of \eqref{eq:theta0bd}. We obtain \eqref{eq:singcon} in consequence and
	therefore conditions (c)(i,ii) both hold. 
	
	Now if \eqref{eqn:theta0b} holds then Eq.~\eqref{eq:theta0bd} holds
	for some $K>0$ because the former is just the $K\to 0+$ limit of the latter. Therefore 
	we have conditions~\eqref{eqn:singconb} and~\eqref{eq:theta0bd} and the previous paragraph
	shows that (c)(i,ii) both hold.
	
	\end{proof}

We now apply this theorem to the more specific setting of Einstein--Klein--Gordon theory, using the worldline bound of Theorem \ref{the:ekgineq}
\begin{corollary}\label{cor:sing}
	Let $(M,g,\phi)$ be a solution to the Einstein--Klein--Gordon equation in dimension $n>2$ and with coupling $\xi\in[0,\xi_c]$. Suppose that $(M,g)$ is globally hyperbolic, let $S$ be a smooth
	spacelike Cauchy surface for $(M,g)$ and suppose that $\phi$ obeys global bounds
	$|\phi|\le\phi_{\text{max}}< (8\pi\xi)^{-1/2}$ and $|\nabla_{\dot{\gamma}}\phi|\le\phi'_{\textrm{max}}$ along all unit speed timelike geodesics $\gamma$ issuing orthogonally from $S$. Assume hypotheses~(\ref{it:lifetime}) and~(\ref{it:singcon}) of Theorem~\ref{thm:sing}, where $Q$ and $\tilde{Q}$ are given by
	\be\label{eq:QQtdef}
	Q=\frac{32\pi \xi \fmax^2}{1-8\pi \xi \fmax^2} \,, \qquad \tilde{Q}^2=\frac{(1-2\xi)m^2}{4\xi(n-2)}+\left(\frac{8\pi \xi \fmax \fmax'}{1-8\pi \xi \fmax^2}\right)^2  \,.
	\ee 
	Then $(M,g)$ is future geodesically incomplete.
\end{corollary}
\begin{proof}
	First, we estimate the right-hand side of Eq.~\eqref{eqn:ffactor}, noting first that $\phi^2/(1-8\pi\xi\phi^2)$ is increasing with $\phi$, and therefore
	\begin{equation}
	\int_\gamma d\tau \left(\frac{1-2\xi}{n-2}\right) \frac{8\pi m^2 \phi^2 f^2(\tau)}{1-8\pi \xi \phi^2} \le \left(\frac{1-2\xi}{n-2}\right) \frac{8\pi m^2 \fmax^2}{1-8\pi \xi \fmax^2}\|f\|^2 =\frac{Qm^2}{4\xi} \left(\frac{1-2\xi}{n-2}\right)\|f\|^2
	\end{equation}  
	for all smooth compactly supported $f$. Next,
	using the inequality
	\be
	||(fg)'||^2 \leq 2( ||f'||^2 ||g||^2_\infty +||f||^2 ||g'||^2_\infty) \,,
	\ee
	and also the global bound on $\phi$, we can write  
	\be
	\int_\gamma d\tau \left(\frac{d}{d\tau}\frac{f(\tau)}{\sqrt{1-8\pi \xi \phi^2}}\right)^2 16\pi\xi\phi^2 \leq \frac{32\pi\xi\fmax^2}{1-8\pi \xi \fmax^2} \left( ||f'||^2+||f||^2 \left(\frac{8\pi \xi \fmax \fmax'}{1-8\pi \fmax^2}\right)^2 \right) 
	\ee 
	for all smooth compactly supported $f$.
	Thus Eq.~(\ref{eqn:ffactor}) implies
	that Eq.~\eqref{eq:Rbdsing} holds with $Q$ and $\tilde{Q}$ given by Eq.~\eqref{eq:QQtdef}. This supplies hypothesis~\eqref{it:Rbdsing} of Theorem~\ref{thm:sing} and as we also assume hypotheses~\eqref{it:lifetime} and~\eqref{it:singcon} the result follows.
\end{proof}
	
As in Theorem~\ref{thm:sing}, one could replace the hypotheses~\eqref{eq:singcon} and~\eqref{eq:theta0bd} by the alternatives Eqs.~\eqref{eqn:singconb} and~\eqref{eqn:theta0b}.

\section{Discussion}
\label{sec:disc}

This paper has accomplished two main goals. First, we have established worldline and worldvolume bounds on the effective energy density of the nonminimally coupled scalar field. Elsewhere, these bounds will be used as the basis for a quantum energy inequality variant of the strong energy condition for the quantized field. Second, we have shown that our bounds can be used to derive
a Hawking-type singularity theorem for the Einstein--Klein--Gordon theory, by applying methods
developed in~\cite{Fewster:2010gm}. The overall message here is that sufficient initial contraction on a compact Cauchy surface is enough to guarantee timelike geodesic incompleteness, even though the non-minimally coupled Klein--Gordon theory does not obey the SEC. 

To conclude, we discuss the hypotheses of the singularity theorem in more depth, aiming for an understanding the physical magnitude of the initial contraction required. Reinserting units and thus 
restoring $G$ and $c$, the constants $Q$ and $\tilde{Q}$ become
\be\label{eq:QQtdefdim}
Q=\frac{32\pi \xi G\fmax^2/c^4}{1-8\pi \xi G\fmax^2/c^4} \,, \qquad \tilde{Q}^2=\frac{(1-2\xi)(mc)^2}{4\xi(n-2)}+\left(\frac{8\pi \xi G\fmax \fmax'/c^4}{1-8\pi \xi G\fmax^2/c^4}\right)^2  \,,
\ee 
where $m$ has units of inverse length. Thus $Q$ is dimensionless, while $\tilde{Q}$ has dimensions of inverse time, as required for dimensional correctness in Eq.~\eqref{eq:theta0bd} with $\tau_0$ having dimensions of time and consequently
$K$ being an inverse time. Evidently both $Q$ and $\tilde{Q}$ become very large if
$\fmax^2$ is allowed to be close to the critical value $c^4/(8\pi G\xi)$. However,
$Q\le 4$ if $\fmax^2$ does not exceed half the critical value, for example, so it is no great restriction to take $Q$ of order $1$. Turning to $\tilde{Q}$, the second term is $Q^2(\fmax'/\fmax)^2$ up to numerical factors. The ratio $\fmax'/(mc\fmax)$
is dimensionless and it would be reasonable to assume bounds so that this quantity does not 
greatly exceed unity. Therefore $\tilde{Q}$ would be reasonably expected not greatly to exceed $mc$. 
The remaining ingredient in the contraction bound Eq.~\eqref{eq:theta0bd} are the timescale $\tau_0$ and constant $K$, which depend 
on the history of the solution prior to $S$; we may assume for the purposes of
discussion that the corresponding terms in Eq.~\eqref{eq:theta0bd}  are not large. 
Accordingly, the initial contraction required to ensure geodesic incompleteness might be expected to be of the order of $mc$, i.e., the characteristic frequency
of the (minimally coupled) Klein--Gordon operator (recall that $m$ is an inverse length in this discussion).

For the purely classical Einstein--Klein--Gordon theory, this seems very reasonable. But---with a view to semiclassical quantum gravity---what if the scalar field is supposed to describe an elementary particle, with a correspondingly small mass? In this situation $m$ is replaced in Eq.~\eqref{eq:QQtdefdim} by $mc/\hbar$, the inverse Compton length for a particle of mass $m$, so $\tilde{Q}$ is of the order of the inverse Compton time and our previous
reasoning would suggest that a huge initial contraction would be required to guarantee geodesic incompleteness. For example, with the physical values of $G$, $\hbar$ and $c$ in $n=4$ dimensions, and with $m$ on the order of the pion mass, $m=140\textrm{MeV}/c^2$, the initial contraction would be of the order $2\times 10^{23}{\textrm{s}}^{-1}$ (using the value $\hbar=6.6\times 10^{-22}\textrm{MeV.s}$) by such arguments. This would call into question the utility of the singularity result for this situation.

However, the value of $\fmax$ should be reconsidered in this hybrid model.
To indicate the values that would be reasonable,
we consider a quantized scalar field in Minkowski spacetime of dimension $n$, in a thermal state of temperature $T<T_m$, where $T_m=mc^2/k$ sets a natural temperature scale for the model, beyond which its applicability might be doubtful. Here, $k$ is Boltzmann's constant. In this regime,
the expectation value of the Wick square is
\begin{equation}
\label{eqn:estim}
\langle {:}\phi^2{:}\rangle_T \sim  C_n \hbar c\left(\frac{kT_m}{\hbar c}\right)^{n-2}
(T/T_m)^{(n-2)/2} K_{(n-2)/2}(T_m/T) ,
\end{equation}
where the numerical constant is 
\begin{equation}
C_n=\frac{S_{n-2}2^{(n-2)/2}}{(2\pi)^{n-1}\sqrt{\pi}} \Gamma((n-1)/2)
\end{equation} 
and takes the value $C_4=0.05$ in $n=4$ dimensions, for example. 
Therefore 
\begin{equation}
\frac{G\langle {:}\phi^2{:}\rangle_T}{c^4}\sim  C_n\left(\frac{T_m}{T_{\text{Pl}}}\right)^{n-2}(T/T_m)^{(n-2)/2} K_{(n-2)/2}(T_m/T),
\end{equation}
where $T_{\text{Pl}}$ is the Planck temperature (in $n$ dimensions). The derivation of this estimate is given in Appendix \ref{app:temp}.
If we take $\fmax^2\sim\langle {:}\phi^2{:}\rangle_T$ then 
the factor $Q$ is reasonably assumed to be given by
\begin{equation}
Q\sim  (m/m_{\textrm{Pl}})^{n-2}(T/T_m)^{(n-2)/2} K_{(n-2)/2}(T_m/T)
\end{equation}
where $m_{\textrm{Pl}}$ is the Planck mass. Maintaining the previous expectation that
$\fmax'/\fmax\sim mc$, and using $T$ to parameterise the maximum acceptable field amplitude in this way leads to a contraction estimate
\be
\theta_0 \lesssim -\frac{mc^2}{\hbar}Q^{1/2} \sim -\frac{mc^2}{\hbar} (m/m_{\textrm{Pl}})^{(n-2)/2}  (T/T_m)^{(n-2)/4} (K_{(n-2)/2}(T_m/T) )^{1/2} \,,
\ee
to guarantee geodesic incompleteness.
For a pion in $n=4$, the leading factor on the right-hand side $m/m_{\textrm{Pl}}=5.9 \times 10^{-20}$, while $mc^2/ \hbar$, as we mentioned, is $2\times 10^{23} \textrm{s}^{-1}$ and $T_m=1.7 \times 10^{12}\textrm{K}$. However, the remaining factor decays very rapidly as $T/T_m$ is reduced below $1$. For example, if $T=10^{-2}T_m$ then $(T/T_m)^{(n-2)/4} (K_{(n-2)/2}(T_m/T) )^{1/2}=6.8 \times 10^{-24}$, thus bringing the overall contraction needed down to the order of $10^{-19} \textrm{s}^{-1}$. This is indeed very small: If a volume were subject to contraction at this constant fractional rate, it would require approximately 100 times the current age of the universe to halve its size. Our calculation has allowed for a maximum temperature scale of $T=10^{10}\textrm{K}$ -- the temperature of the universe approximately one second after the Big Bang. Repeating the calculation for the Higgs mass $125\textrm{GeV}/c^2$, 
a minimum contraction of order $10^{-14}\textrm{s}^{-1}$ is required, but this time allowing 
a temperature up to $T=10^{13}\textrm{K}$, the temperature of the Universe  at age $0.0001\textrm{s}^{-1}$.

Summarising this discussion: a model in which the field mass is taken equal to 
an elementary particle mass would need very little initial contraction to guarantee that either there is geodesic incompleteness
or that, at the least, the solution evolves to a situation where the natural energy scales associated with the field approach those of the early universe. At this stage, a macroscopic observer might be forgiven for believing that a singularity had occurred!

Hawking and Ellis \cite{HawkingEllis:1973} discuss the violation of an average SEC by a pion. Their heuristic analysis led them to argue that the convergence of timelike geodesics would not be influenced by SEC violation on scales greater than $10^{-14} \textrm{m}$. By contrast, our analysis instead rigorously shows that even for extremely small initial contractions a singularity is inevitable. 

The obvious extension of this work is the derivation of a quantum strong energy inequality for the non-minimally coupled field. To examine if such an inequality can be used to prove a Hawking-type singularity theorem, it is also necessary to find estimates on the timescales for averaging over which the curved spacetime approximates the Minkowski results. 

It should be noted that to fully understand whether the dynamics are driven towards
singularities at the semi-classical level requires a semiclassical analysis that takes backreaction into account in a dynamical way. Positive results include \cite{Drago:2014eoa}, which calculates geometrical fluctuations from (passive) quantum fields, and a result on ANEC with transverse smearing \cite{Flanagan:1996gw}. However, the calculation of backreaction, even in the second order in perturbation theory brings significant technical challenges to the problem. Finally, the question of whether full quantum gravity can resolve singularities via a ``big bounce'' or not remains open.

\section*{Acknowledgments}
We thank Roger Colbeck for comments on the text. This work is part of a project that has received funding from the European Union's Horizon 2020 research and innovation programme under the Marie Sk\l odowska-Curie grant agreement No. 744037 ``QuEST''. P.J.B.\ thanks the WW Smith Fund for their financial support.

\appendix

\section{Temperature dependence of Wick square}
\label{app:temp}

 Consider the KMS state of the Klein--Gordon field with mass $m$ in $n$-dimensional Minkowski space for $n> 3$ (or $m>0$ if $n=3$) corresponding to a thermal equilibrium state at temperature $T$. The two-point function of that state is
\be
W^{(2)}_T(t,\mathbf{x},t',\mathbf{x}')=\int d\mu(\mathbf{k}) \left( \frac{e^{-i((t-t')\omega(\mathbf{k})-(\mathbf{x}-\mathbf{x}') \mathbf{k})}}{1-e^{-\beta\omega(\mathbf{k})}}+\frac{e^{i((t-t')\omega(\mathbf{k})-(\mathbf{x}-\mathbf{x}') \mathbf{k})}}{e^{\beta\omega(\mathbf{k})}-1}  \right) \,,
\ee
where $\beta=(kT)^{-1}$ and $\mu(\mathbf{k})$ is given by 
\be
\label{eqn:measure}
d\mu(\bfk)=\int \frac{d^{n-1}\bfk }{(2\pi)^{n-1}}\frac{1}{2\omega(\bfk)} \,,
\ee
 with $\omega(\mathbf{k})=\sqrt{\mathbf{k}^2+m^2}$. Here we use $\hbar=c=1$. After renormalizing with the ground state two-point function
 \be
 \label{eqn:vac}
 W^{(2)}_0(t,\mathbf{x},t',\mathbf{x}')=\int d\mu (\mathbf{k}) e^{-i[(t-t')\omega(\mathbf{k})-(\mathbf{x}-\mathbf{x}')\mathbf{k}]} \,,
 \ee
 in the coincidence limit we find
\be
\langle {:}\phi^2{:}\rangle_T =[ W^{(2)}_T - W^{(2)}_0]_c
= \frac{2 S_{n-2}}{(2\pi)^{n-1}} \int_m^\infty d\omega (\omega^2-m^2)^{(n-3)/2} \frac{1}{e^{\beta \omega(k)}-1} \,,
\ee
where $S_{n-2}$ is the area of the unit $(n-2)$-sphere. Changing variables to $x= \omega/m$ gives
\be \label{eqn:betaone}
\langle {:}\phi^2{:}\rangle_T=\frac{2 S_{n-2}m^{n-2}}{(2\pi)^{n-1}} \int_1^\infty dx (x^2-1)^{(n-3)/2} \frac{1}{e^{\beta x/m}-1}  \,.
\ee 
Reinserting units gives
	\begin{equation}
	\langle {:}\phi^2{:}\rangle_T =  \hbar c\frac{S_{n-2}(k T_m/(\hbar c))^{n-2}}{(2\pi)^{n-1}}\int_{1}^\infty dx\frac{(x^2-1)^{(n-3)/2}}{e^{T_m x/T}-1} \,.
	\end{equation}
	
	As we are interested in $T<T_m$, we may expand the integrand using a geometric series 
	\begin{equation}
	\langle {:}\phi^2{:}\rangle_T = \hbar c\frac{S_{n-2}(k T_m/(\hbar c))^{n-2}}{(2\pi)^{n-1}}\int_{1}^\infty dx \sum_{r=1}^\infty (x^2-1)^{(n-3)/2}e^{-rT_m x/T}\,,
	\end{equation}
	and exchanging sum and integral (which is valid) one obtains
	\begin{equation}
	\langle {:}\phi^2{:}\rangle_T =\hbar c \frac{S_{n-2}(k T_m/(\hbar c))^{n-2}\Gamma((n-1)/2)}{(2\pi)^{n-1}\sqrt{\pi}}  \sum_{r=1}^\infty 
	\left(\frac{2T}{rT_m}\right)^{(n-2)/2} K_{(n-2)/2}(rT_m/T) \,,
	\end{equation}
	in terms of modified Bessel functions. For order of magnitude purposes, the first term in the sum is perfectly adequate for $T<T_m$. This becomes more obvious if we subtract the first term from the sum
	\be
	\frac{1}{e^{T_m x/T}-1}-\frac{1}{e^{T_m x/T}} \leq 2 e^{-2 T_m x/T} 
	\,,
	\ee
	which is double the second term of the series. This gives the estimate of Eq.~\eqref{eqn:estim}.

\bibliography{class}

\end{document}